\pdfoutput=1

\RequirePackage[l2tabu,orthodox]{nag}

\documentclass[11pt,letterpaper,reqno,oneside]{article}

\DeclareTextFontCommand{\emph}{\slshape}

\makeatletter
\renewcommand{\paragraph}{%
	\@startsection{paragraph}{4}%
	{\z@}{1.75ex \@plus 1ex \@minus .2ex}{-0.7em}%
	{\normalfont\normalsize\bfseries}%
}
\makeatother

\usepackage[T1]{fontenc}
\usepackage{lmodern}
\usepackage[utf8]{inputenc}

\usepackage[american,german,british]{babel}

\usepackage[activate={true,nocompatibility}]{microtype}

\usepackage{csquotes}

\usepackage[backend=biber,autolang=other,style=apa,maxcitenames=5,sortcites=false]{biblatex}
\DeclareLanguageMapping{british}{british-apa}
\DeclareLanguageMapping{american}{american-apa}

\usepackage{amsmath}
\usepackage{amssymb}
\usepackage{amsfonts}
\usepackage{amsthm}
\usepackage{mathtools}
\usepackage{stmaryrd}
\usepackage{centernot}

\let\originalleft\left
\let\originalright\right
\renewcommand{\left}{\mathopen{}\mathclose\bgroup\originalleft}
\renewcommand{\right}{\aftergroup\egroup\originalright}

\usepackage{graphicx}
\usepackage{tikz,pgfplots}
\pgfplotsset{compat=1.10}
\usetikzlibrary{shapes.multipart}
\usetikzlibrary{decorations.markings}

\usepackage[title]{appendix}

\usepackage{datetime}
\newdateformat{datestyle}{ {\THEDAY} \monthname[\THEMONTH] \THEYEAR }

\usepackage{enumerate}
\usepackage{enumitem}
\setlist[enumerate,1]{label=(\arabic*)}
\setlist[itemize,1]{label=--}
\setlist[itemize,2]{label=--}
\setlist[itemize,3]{label=--}
\setlist[itemize,4]{label=--}

\usepackage{xcolor}
\usepackage{verbatim}
\usepackage{gensymb}
\usepackage{rotating}
\usepackage{subcaption}
\usepackage{floatpag}
\floatpagestyle{plain}
\usepackage{marvosym}

\usepackage{hyphenat}
\theoremstyle{definition}

\newtheorem{theorem}{Theorem}
\newtheorem{proposition}{Proposition}
\newtheorem{lemma}{Lemma}
\newtheorem{corollary}{Corollary}
\newtheorem{remark}{Remark}
\newtheorem{observation}{Observation}
\newtheorem{definition}{Definition}

\newtheoremstyle{named}
	{\topsep}					
	{\topsep}					
	{}							
	{0pt}						
	{\bfseries}					
	{}							
	{5pt plus 1pt minus 1pt}	
	{\thmnote{#3}}				
\theoremstyle{named}
\newtheorem{namedthm}{}

\usepackage{xpatch}
\xpatchcmd{\proof}{\itshape}{\proofheadfont}{}{}
\newcommand{\proofheadfont}{\slshape}

\usepackage{hyperref}
\hypersetup{
	colorlinks=true,
	linkcolor=black,
	citecolor=black,
	filecolor=black,
	urlcolor=black,
}

\usepackage[nameinlink]{cleveref}
\crefname{page}{p.}{pp.}
\crefname{equation}{equation}{equations}
\crefname{section}{section}{sections}
\crefname{subsection}{section}{sections}
\crefname{subsubsection}{section}{sections}
\crefname{appsec}{appendix}{appendices}
\crefname{supplsec}{supplemental appendix}{supplemental appendices}
\crefname{footnote}{footnote}{footnotes}
\crefname{figure}{figure}{figures}
\crefname{table}{table}{tables}
\crefname{theorem}{theorem}{theorems}
\crefname{proposition}{proposition}{propositions}
\crefname{lemma}{lemma}{lemmata}
\crefname{corollary}{corollary}{corollaries}
\crefname{remark}{remark}{remarks}
\crefname{observation}{observation}{observations}
\crefname{example}{example}{examples}
\crefname{fact}{fact}{facts}
\crefname{definition}{definition}{definitions}
\crefname{assumption}{assumption}{assumptions}
\crefname{exercise}{exercise}{exercises}
\crefname{notation}{notation}{notation}
\crefname{claim}{claim}{claims}
\crefname{conjecture}{conjecture}{conjectures}

\newcommand{\dd}{\mathrm{d}}

\newcommand*{\xslant}[2][76]{%
	\begingroup
	\sbox0{#2}%
	\pgfmathsetlengthmacro\wdslant{\the\wd0 + cos(#1)*\the\wd0}%
	\leavevmode
	\hbox to \wdslant{\hss
		\tikz[
			baseline=(X.base),
			inner sep=0pt,
			transform canvas={xslant=cos(#1)},
		] \node (X) {\usebox0};%
		\hss
		\vrule width 0pt height\ht0 depth\dp0 %
	}%
	\endgroup
}
\makeatletter
\newcommand*{\xslantmath}{}
\def\xslantmath#1#{%
	\@xslantmath{#1}%
}
\newcommand*{\@xslantmath}[2]{%
	\ensuremath{%
		\mathpalette{\@@xslantmath{#1}}{#2}%
	}%
}
\newcommand*{\@@xslantmath}[3]{%
	\xslant#1{$#2#3\m@th$}%
}
\makeatother

\makeatletter
\def\namedlabel#1#2{\begingroup
	#2%
	\def\@currentlabel{#2}%
	\phantomsection\label{#1}\endgroup
}
\makeatother

\makeatletter
\let\save@mathaccent\mathaccent
\newcommand*\if@single[3]{%
	\setbox0\hbox{${\mathaccent"0362{#1}}^H$}%
	\setbox2\hbox{${\mathaccent"0362{\kern0pt#1}}^H$}%
	\ifdim\ht0=\ht2 #3\else #2\fi
	}
\newcommand*\rel@kern[1]{\kern#1\dimexpr\macc@kerna}
\newcommand*\widebar[1]{\@ifnextchar^{{\wide@bar{#1}{0}}}{\wide@bar{#1}{1}}}
\newcommand*\wide@bar[2]{\if@single{#1}{\wide@bar@{#1}{#2}{1}}{\wide@bar@{#1}{#2}{2}}}
\newcommand*\wide@bar@[3]{%
	\begingroup
	\def\mathaccent##1##2{%
	  \let\mathaccent\save@mathaccent
	  \if#32 \let\macc@nucleus\first@char \fi
	  \setbox\z@\hbox{$\macc@style{\macc@nucleus}_{}$}%
	  \setbox\tw@\hbox{$\macc@style{\macc@nucleus}{}_{}$}%
	  \dimen@\wd\tw@
	  \advance\dimen@-\wd\z@
	  \divide\dimen@ 3
	  \@tempdima\wd\tw@
	  \advance\@tempdima-\scriptspace
	  \divide\@tempdima 10
	  \advance\dimen@-\@tempdima
	  \ifdim\dimen@>\z@ \dimen@0pt\fi
	  \rel@kern{0.6}\kern-\dimen@
	  \if#31
	    \overline{\rel@kern{-0.6}\kern\dimen@\macc@nucleus\rel@kern{0.4}\kern\dimen@}%
	    \advance\dimen@0.4\dimexpr\macc@kerna
	    \let\final@kern#2%
	    \ifdim\dimen@<\z@ \let\final@kern1\fi
	    \if\final@kern1 \kern-\dimen@\fi
	  \else
	    \overline{\rel@kern{-0.6}\kern\dimen@#1}%
	  \fi
	}%
	\macc@depth\@ne
	\let\math@bgroup\@empty \let\math@egroup\macc@set@skewchar
	\mathsurround\z@ \frozen@everymath{\mathgroup\macc@group\relax}%
	\macc@set@skewchar\relax
	\let\mathaccentV\macc@nested@a
	\if#31
	  \macc@nested@a\relax111{#1}%
	\else
	  \def\gobble@till@marker##1\endmarker{}%
	  \futurelet\first@char\gobble@till@marker#1\endmarker
	  \ifcat\noexpand\first@char A\else
	    \def\first@char{}%
	  \fi
	  \macc@nested@a\relax111{\first@char}%
	\fi
	\endgroup
}
\makeatother

\newcommand{\MPS}{\mathrel{M}}
\newcommand{\nMPS}{\mathrel{\centernot{\mathrel{\MPS}}}}

\title{\scshape Statistical discrimination and statistical informativeness\thanks{We thank Chris Chambers for useful comments.}}

\author{%
\begin{tabular}{cc}
	Matteo Escudé
	& Paula Onuchic \\
	LUISS
	& University of Oxford \vspace{1em} \\
	\parbox{\widthof{Quitzé Valenzuela-Stookey}}{\centering Ludvig Sinander}
	& Quitzé Valenzuela-Stookey \\
	University of Oxford
	& Duke University
\end{tabular}%
}

\date{31 May 2022}

\makeatletter
	\AtBeginDocument{ \hypersetup{ 
		pdftitle = {Statistical discrimination and statistical informativeness}, 
		pdfauthor = {Matteo Escudé, Paula Onuchic, Ludvig Sinander and Quitzé Valenzuela-Stookey} 
		} }
\makeatother

\begin{document}

\maketitle

\begin{abstract}
We study the link between Phelps--Aigner--Cain-type statistical discrimination and familiar notions of statistical informativeness. 
Our central insight is that Blackwell's Theorem, suitably relabeled, characterizes statistical discrimination in terms of statistical informativeness.
This delivers one-half of Chambers and Echenique's (\citeyear{ChambersEchenique2021}) characterization of statistical discrimination as a corollary, and suggests a different interpretation of it: that discrimination is inevitable.
In addition, Blackwell's Theorem delivers a number of finer-grained insights into the nature of statistical discrimination.
We argue that the discrimination--informativeness link is quite general, illustrating with an informativeness characterization of a different type of discrimination.
\end{abstract}

\section{Introduction}

\emph{Statistical discrimination} is the umbrella term for a number of interrelated information-economic theories that seek to understand why individuals face discriminatory treatment based on their social identities.
An important strand began with \textcite{Phelps1972} and \textcite{AignerCain1977}, who showed that in a competitive labor market, wage-discrimination can arise even if firms care only about workers' skills (not their identities) and know that skill and identity are statistically independent.\footnote{For an overview of this literature,
	including an account of how the Phelps and Aigner--Cain papers differ, see the surveys by \textcite{FangMoro2011} and \textcite{Onuchic2022}.}
The reason is that for a firm that can observe a worker's skill only imperfectly, identity may be an informative signal about skill, \emph{despite} the fact that identity is uninformative about skill ex ante.
Consequently, two populations that differ only in respect of social identity (not skill) may be paid differently on average.

In a recent paper, \textcite{ChambersEchenique2021} develop a characterization of when and why Phelps--Aigner--Cain-type statistical discrimination arises, in a general setting which eschews the functional-form assumptions on information and payoffs that typify the earlier literature.
This generality is made possible by a conceptual innovation: the authors pose their question in the ``belief-based'' language of the recent literature on information design.\footnote{See \textcite{Kamenica2019} for a survey.}

Chambers and Echenique take as primitive the set of possible observable traits (or ``signals'') that a worker might have, which are statistically informative about her skill.
A worker's trait might comprise her standardized test scores, her CV, and a reference letter, for example.
The authors ask which sets of traits are \emph{non-discriminatory,} in the sense that any two populations with the same skill distribution, but potentially different distributions of traits, are paid the same on average.
Their main result provides an answer.

In this paper, we revisit the question of when and why Phelps--Aigner--Cain statistical discrimination occurs, using Chambers and Echenique's ``belief-based'' formalism. Our key insight is that Blackwell's Theorem, suitably relabeled, characterizes statistical discrimination as arising precisely from the statistical informativeness about skill of a population's observable traits.

We draw several implications.
First, Blackwell's Theorem delivers one-half of Chambers and Echenique's main result as a corollary, and suggests a different interpretation of it: that statistical discrimination is \emph{inevitable.}
Secondly, Blackwell delivers additional, finer-grained insights into how, when and why discrimination occurs.

We further argue that the tight link between statistical discrimination and statistical informativeness embodied by Blackwell's Theorem is quite general, extending beyond the competitive wage-setting environment.
To illustrate, we consider another kind of discrimination:
we alter the model so that a firm can observe a worker's trait only at a cost (by ``interviewing'' her), and show how \emph{ex-ante discrimination} in the form of exclusion from interviews is characterized by statistical informativeness.
Depending on the magnitude of the cost,
the relevant notion of informativeness here is either Blackwell's or that of \textcite{LehrerShmaya2008}.

Beyond clarifying, and broadening, the link between statistical informativeness and various notions of statistical discrimination, this paper contributes to the literature by applying a well-known result in information economics to a workhorse model of discrimination. Our insight is that this standard result (Blackwell's Theorem) delivers a sharp characterization of statistical discrimination.
Our analysis suggests that old and new results in information economics may be harnessed to deepen our understanding of economic theories of discrimination.

\section{Environment}
\label{sec:env}

The setting is that of \textcite{ChambersEchenique2021}, with their notation.
It is a model in the spirit of \textcite{Phelps1972,AignerCain1977}.

Firms are exogenously matched to workers,
and each worker is paid the expected surplus that she generates for her employer.
This surplus depends both on the worker's unobservable skill and on how those skills match up with the firm's technology.
Formally, each worker has a \emph{skill type} $\theta$ drawn from a non-empty finite set $\Theta$ of possible skill types,
and a \emph{firm} is a non-empty finite set $A\subset\mathbb{R}^\Theta$.
The interpretation is that each $a \in A$ is a task to which the firm can assign a worker, and that assigning a worker of skill type $\theta \in \Theta$ to task $a$ yields surplus $a(\theta)\in\mathbb{R}$.

Since worker skill is unobservable, a firm must base its estimate of a worker's skill on that worker's observable trait (or ``signal'').
A worker's trait may comprise her standardized test scores, her reference letters, and her personal charm, for example.
Given the joint distribution of skill types and observable traits in the population from which the worker is drawn, each possible worker trait induces a Bayesian posterior belief $s \in \Delta(\Theta)$ about worker skill.
Rather than model traits directly, \textcite{ChambersEchenique2021} employ the ``belief-based'' formalism of the information design literature:
they identify each trait with the posterior belief $s \in \Delta(\Theta)$ about skill that it induces.
In other words, they label each trait with the posterior that it gives rise to.

The set of possible traits is a measurable subset $\mathcal{S}$ of $\Delta(\Theta)$.
A \emph{population (of workers)} is a distribution $\pi\in\Delta(\mathcal{S})$ of traits. The interpretation is that $\pi(S) \in [0,1]$ is the fraction of workers in this population whose observable trait belongs to $S \subset \mathcal{S}$, for each (measurable) set $S \subset \mathcal{S}$.
The distribution of skill types in a population $\pi \in \Delta(\mathcal{S})$ is denoted $p_\pi \in \Delta(\Theta)$; it is given by
\begin{equation*}
	p_\pi(\theta) \coloneqq \int_\mathcal{S}s(\theta)\pi(\dd s)
	\quad \text{for each $\theta \in \Theta$.}
\end{equation*}
That is, $p_\pi(\theta)$ is the fraction of workers in the population $\pi$ whose skill is $\theta$.

Observe that populations can differ both in their underlying skill distributions and in how informative their observable traits are about skill. That is, two populations $\pi,\pi' \in \Delta(\mathcal{S})$ may have different skill distributions $p_\pi \neq p_{\pi'}$, and even if not, they may still differ ($\pi \neq \pi'$) because their distributions of observable traits are differentially informative about skill.
To illustrate, consider two populations that have taken different standardized tests: the SAT and the ACT, respectively.
Differences between these populations in the distribution of a Bayesian firm's beliefs about members' skill types may arise both from differences in the populations' actual skill distributions and from differences in the statistical informativeness about skill of the SAT and ACT. Alternatively, differences in a firm's beliefs about members' skill types can arise even if the populations take the same test, provided there is some characteristic of the population that affects test outcomes.

Finally, recall that a \emph{firm} is a non-empty finite set $A \subset \mathbb{R}^\Theta$ of tasks. Firm $A$ assigns its worker to whichever task $a \in A$ yields the highest expected surplus, given its information about the worker's skill type.
Thus the expected surplus of a firm $A$ faced with a worker whose observable trait is $s \in \mathcal{S}$ is
\begin{equation*}
	v_A(s) \coloneqq \max_{a\in A}
	\sum_{\theta \in \Theta} a(\theta) s(\theta) .
\end{equation*}
Each worker is paid the expected surplus that she generates.

\section{The Chambers--Echenique result}
\label{sec:ce}

In this section, we state the first half of Chambers and Echenique's (\citeyear{ChambersEchenique2021}) main result, and suggest an alternative, dismal interpretation of it.

\begin{definition}[non-discriminatory]
\label{def:nondisc}
A set $\mathcal{S}$ of traits is \emph{non-discriminatory} if
for any populations $\pi, \pi'\in\Delta(\mathcal{S})$
and any firm $A \subset \mathbb{R}^\Theta$,
if $p_\pi=p_{\pi'}$, then
$\int_\mathcal{S}v_A \dd \pi=\int_\mathcal{S}v_A \dd \pi'$.
\end{definition}

The essence of this definition is that an environment is free from discrimination whenever any two populations with the same skill distribution receive, on average, the same remuneration.
Formally, the definition requires this to hold \emph{at every firm,} separately.

\begin{definition}[identification]
A set $\mathcal{S}$ of traits is \emph{identified} if for any $\pi,\pi'\in\Delta(\mathcal{S})$, $p_\pi=p_{\pi'}$ implies that $\pi=\pi'$.
\end{definition}

\textcite{ChambersEchenique2021} give an ``econometric'' interpretation of identification:
it requires that the ``parameter'' $\pi$ be recoverable from the ``observable magnitude'' $p_\pi$.

\begin{theorem}[Chambers--Echenique]
\label{th:1}
A set $\mathcal{S}$ of traits is non-discriminatory if and only if it is identified.
\end{theorem}

This result is one-half of Chambers and Echenique's (\citeyear{ChambersEchenique2021}) Theorem 1;
the remaining half asserts that these two properties are equivalent to a third property, on which we shall not comment.

\subsection{A dismal interpretation: inevitable discrimination}
\label{sec:ce:identification}

We suggest the following alternative interpretation of ``identification'': there \emph{are} no distinct populations $\pi,\pi' \in \Delta(\mathcal{S})$ with the same skill distribution.

On our interpretation, the ``if'' part of \Cref{th:1} is clear:
identification implies that non-discrimination holds \emph{vacuously,} because there simply aren't any populations between whom discrimination could occur according to \Cref{def:nondisc}.
(That is: under identification, any two populations $\pi,\pi'$ with the same skill distribution $p_\pi = p_{\pi'}$ must be one and the same population $\pi = \pi'$, so trivially receive the same remuneration.)

The ``only if'' part of \Cref{th:1} therefore says precisely that discrimination is inevitable:
whenever there \emph{are} distinct populations with the same skill distribution, there will be discrimination between them.

\section{Discrimination and informativeness}
\label{sec:blackwell}

In this section, we state Blackwell's Theorem with terms relabeled, and show that it entails \Cref{th:1} as well as further insights into discrimination.

\begin{definition}
	\label{definition:disc}
	For populations $\pi,\pi' \in \Delta(\mathcal{S})$ with the same skill distribution ($p_\pi=p_{\pi'}$):
	\begin{enumerate}
	
		\item Firm $A \subset \mathbb{R}^\Theta$ \emph{discriminates (strictly) against} population $\pi$ if
		\begin{equation*}
			\int_{\mathcal{S}} v_A \dd \pi 
			\leqslant\; \mathrel{(<)} \int_{\mathcal{S}} v_A \dd \pi' .
		\end{equation*}	

		\item There is \emph{systematic discrimination} against population $\pi$ if every firm discriminates against $\pi$,
		and some firm discriminates strictly against $\pi$.

		\item There is \emph{unsystematic discrimination} if there are firms $A,A' \subset \mathbb{R}^\Theta$ such that firm $A$ ($A'$) discriminates strictly against $\pi$ ($\pi'$).

		\item There is \emph{no discrimination} if there is neither systematic nor unsystematic discrimination.
	
	\end{enumerate}
\end{definition}

Write $\pi' \MPS \pi$ when $\pi'$ is a mean-preserving spread of $\pi$.\footnote{That is: $\int_{\mathcal{S}} h \dd \pi' \geqslant \int_{\mathcal{S}} h \dd \pi$ for every convex and continuous function $h: \Delta(\Theta) \rightarrow \mathbb{R}$.}
Intuitively, $\pi' \MPS \pi$ means that $\pi'$ is ``more spread out'' than $\pi$, and thus more informative.

\begin{namedthm}[Blackwell's Theorem.]
	\label{theorem:Blackwell}
	For two populations $\pi,\pi' \in \Delta(\mathcal{S})$ with the same skill distribution ($p_\pi=p_{\pi'}$),
	$\pi' \MPS \pi$
	holds if and only if there is either (i) systematic discrimination against $\pi$ or (ii) no discrimination.
\end{namedthm}

This is (one version of) Blackwell's Theorem \parencite{Blackwell1951,Blackwell1953}, with terms relabeled.
The standard labels are as follows:
$\pi,\pi'$ are two \emph{distributions of posterior beliefs} about the unknown \emph{state} $\theta$,
$p_\pi = p_{\pi'}$ is the \emph{prior belief} about the state,
and ``systematic discrimination against $\pi$''
reads ``$\pi'$ is Blackwell strictly more informative than $\pi$.''\footnote{Recall that the Blackwell order is defined as follows.
A \emph{decision problem} is a non-empty finite set $A \subset \mathbb{R}^\Theta$, where $a(\theta)$ is the payoff of action $a$ if the state is $\theta$; its \emph{value} given (posterior) belief $s \in \Delta(\Theta)$ about the state is $v_A(s) = \max_{a\in A} \sum_{\theta \in \Theta} a(\theta) s(\theta)$. We say that $\pi'$ is \emph{Blackwell (strictly) more informative} than $\pi$ exactly if $\pi'$ delivers a weakly greater expected value than $\pi$ in any decision problem (and strictly in some decision problem).}

\begin{corollary}
	\label{corollary:blackwell}
	Consider two populations $\pi,\pi' \in \Delta(\mathcal{S})$ with the same skill distribution ($p_\pi=p_{\pi'}$).
	\begin{enumerate}[label=(\alph*)]
	
		\item \label{item:syst}
		There is systematic discrimination against $\pi$ if and only if $\pi' \MPS \pi \neq \pi'$.

		\item \label{item:unsyst}
		There is unsystematic discrimination if and only if $\pi \nMPS \pi' \nMPS \pi$.

		\item \label{item:nodisc}
		There is no discrimination if and only if $\pi = \pi'$.
	
	\end{enumerate}
\end{corollary}

In proving this corollary below,
we make use of the anti-symmetry of the order $\MPS$.
For completeness, we give a proof of this property in the \hyperref[sec:appendix]{appendix}.

\begin{proof}
	Part \ref{item:unsyst} is exactly (the contra-positive of) \hyperref[theorem:Blackwell]{Blackwell's Theorem}.

	For part \ref{item:nodisc},
	obviously $\pi = \pi'$ implies no discrimination,
	while if there is no discrimination
	then $\pi \MPS \pi' \MPS \pi$ by \hyperref[theorem:Blackwell]{Blackwell's Theorem},
	whence $\pi=\pi'$ since $\MPS$ is anti-symmetric.

	For the ``if'' direction of part \ref{item:syst},
	if $\pi' \MPS \pi \neq \pi'$,
	then by \hyperref[theorem:Blackwell]{Blackwell's Theorem} we have either (i) or (ii),
	and by part \ref{item:nodisc} it cannot be (ii).
	For the ``only if'' direction,
	if there is systematic discrimination against $\pi$
	then $\pi' \MPS \pi$ by \hyperref[theorem:Blackwell]{Blackwell's Theorem},
	and $\pi \neq \pi'$ by part \ref{item:nodisc}.
\end{proof}

\Cref{corollary:blackwell} immediately implies the non-trivial half of \Cref{th:1}:

\begin{proof}[Proof of \Cref{th:1}]
	Suppose that $\mathcal{S}$ is identified.
	Then no two distinct populations $\pi,\pi' \in \Delta(\mathcal{S})$ have $p_\pi=p_{\pi'}$, so non-discrimination holds (vacuously). 

	Suppose that $\mathcal{S}$ is not identified.
	Then there are distinct populations $\pi,\pi' \in \Delta(\mathcal{S})$ with $p_\pi=p_{\pi'}$.
	So by \Cref{corollary:blackwell}\ref{item:nodisc}, there is discrimination.
\end{proof}

Furthermore, whereas \Cref{th:1} tells us only that discrimination is inevitable, 
\Cref{corollary:blackwell} provides more nuanced information.
\emph{Systematic} discrimination, in the form of lower average pay across \emph{all} firms, is experienced precisely by those populations that are less informative, in the sense of $\MPS$. By contrast, populations that are distinct but not $\MPS$-ranked experience mixed fortunes: higher average pay at some firms, and lower average pay at others.

\Cref{corollary:blackwell} also bears on the question of whether statistical discrimination is best understood as arising from differences of informativeness (in the $\MPS$ sense) between populations.
\textcite{ChambersEchenique2021} argue that
``the focus on informativeness in \textcite{Phelps1972,AignerCain1977} is misleading. There may be statistical discrimination even when \dots one population is not more informative than the other.''
\Cref{corollary:blackwell} synthesizes the views on both sides:
\emph{systematic} discrimination occurs precisely against less informative populations (vindicating Phelps and Aigner \& Cain),
while the \emph{unsystematic} discrimination allowed by Chambers and Echenique is a broader (indeed, inevitable) phenomenon.
However, informativeness is the key even in the latter case:
a population experiences discrimination (systematic or not)
exactly if it is \emph{not more informative} in the $\MPS$ sense.

\Cref{corollary:blackwell} may be relevant for policy debates about biases in standardized evaluations. For example, the SAT, which is used for college admissions decisions in the United States, has been criticized for framing questions in a way that favors white students \parencite[e.g.][]{KimZabelina2015}.
If the effect is to increase the statistical informativeness of white students' scores, then \Cref{corollary:blackwell}\ref{item:syst} suggests that systematic discrimination will result.
More generally, \Cref{corollary:blackwell}\ref{item:nodisc} predicts that if the framing of questions affects different populations' test-score distributions differently (in any way), then it will cause discrimination of some form.

\section{The broader connection between discrimination and informativeness}
\label{sec:broad}

The previous section uncovered a tight link between discrimination and informativeness. This link may be used to characterize other forms of statistical discrimination. In this section, we explore one such extension.

The preceding analysis considered discrimination on the intensive margin: all workers are employed, but some populations may receive lower average pay. Another concern is extensive-margin discrimination: workers from some populations may find it harder to obtain employment. We call this \emph{ex-ante discrimination.}

To incorporate ex-ante discrimination into the model studied above, we interpret a firm's observation of a worker's trait as the outcome of a screening or interview process. We augment the model with an ex-ante stage at which a firm observes only to which population a worker belongs, and decides whether to learn the worker's trait by interviewing her. This process is costly, and thus a firm undertakes to interview a worker if and only if the expected value from doing so exceeds the cost. We are concerned with discrimination across populations at this ex-ante stage.

Formally, we extend the above model as follows.
Each firm observes to which population a worker belongs, and can additionally observe the worker's trait at a cost $c > 0$.
Having interviewed a worker, the firm decides whether or not to hire her. If the worker is hired, then she is optimally assigned a task as before, generating some expected surplus;
if she is not hired, then surplus is zero.
To make the firm's interviewing decision non-trivial, we assume that it retains some positive fraction $\alpha \in (0,1]$ of the surplus generated by any worker whom it hires.

A firm is now $(A,\alpha)$, where $A$ is a non-empty finite subset of $\mathbb{R}^\Theta$ as before, and $0 < \alpha \leqslant 1$.
The interview cost $c > 0$ is the same for all firms.
Once a firm has interviewed a worker and learned her trait $s \in \mathcal{S}$, it hires her exactly if $v_A(s) > 0$; thus expected surplus is $\max\{v_A(s),0\}$.
The expected surplus from interviewing a worker belonging to population $\pi \in \Delta(\mathcal{S})$ is thus $\int_{\mathcal{S}} \max\{v_A,0\} \dd \pi$,
so such a worker is interviewed if and only if $\alpha\int_{\mathcal{S}} \max\{v_A,0\} \dd \pi > c$.
Say that a population is \emph{excluded} by firm $(A,\alpha)$ if it is not interviewed. 

\begin{definition}
	\label{definition:hiring_disc}
	Fix the interview cost $c > 0$.
	For populations $\pi,\pi' \in \Delta(\mathcal{S})$ with the same skill distribution ($p_\pi=p_{\pi'}$):
	\begin{enumerate}
	
		\item There is \emph{systematic ex-ante discrimination} against $\pi$ if every firm that excludes $\pi'$ also excludes $\pi$, and some firm excludes $\pi$ but not $\pi'$.

		\item There is \emph{unsystematic ex-ante discrimination} if there are firms that exclude $\pi$ but not $\pi'$
		and firms that exclude $\pi'$ but not $\pi$.

		\item There is \emph{no ex-ante discrimination} if there is neither systematic nor unsystematic ex-ante discrimination.
	
	\end{enumerate}
\end{definition}

\begin{observation}\label{obs:normalization}
Although our definition of ex-ante discrimination quantifies over all firms $(A,\alpha)$, it actually makes no difference if only those with $\alpha=1$ are considered.
That is because for any firm $(A,\alpha)$, the firm $(A',1)$ with $A' = \{\alpha a : a \in A \}$ excludes the same set of populations.
\end{observation}

The following characterizes ex-ante discrimination.

\begin{proposition}
\label{prop:1}
	\label{proposition:lehrer-shmaya}
	Fix an interview cost $c > 0$,
	and consider two populations $\pi,\pi' \in \Delta(\mathcal{S})$ with the same skill distribution ($p_\pi=p_{\pi'}$).
	\begin{enumerate}[label=(\alph*)]
	
		\item \label{item:syst_ls}
		There is systematic ex-ante discrimination against $\pi$ if and only if $\pi' \MPS \pi \neq \pi'$.

		\item \label{item:unsyst_ls}
		There is unsystematic ex-ante discrimination if and only if $\pi \nMPS \pi' \nMPS \pi$.

		\item \label{item:nodisc_ls}
		There is no ex-ante discrimination if and only if $\pi = \pi'$.
	
	\end{enumerate}
\end{proposition}

The dismal conclusion persists: whenever there are distinct populations with the same skill distribution, there will be ex-ante discrimination.
It is notable that \Cref{proposition:lehrer-shmaya}'s characterization of (systematically) ex-ante discriminated-against populations does not depend on the magnitude of the interview cost $c>0$, and that it coincides exactly with the characterization of pay discrimination in \Cref{corollary:blackwell}:
in all cases, discrimination tracks informativeness in the sense of $\MPS$.

The key to \Cref{proposition:lehrer-shmaya} is the following Blackwell-like lemma.
The proposition follows from this lemma by the exact same argument used in \cref{sec:blackwell} to deduce \Cref{corollary:blackwell} from \hyperref[theorem:Blackwell]{Blackwell's Theorem}.

\begin{lemma}
	\label{lemma:LS}
	For two populations $\pi,\pi' \in \Delta(\mathcal{S})$ with the same skill distribution ($p_\pi=p_{\pi'}$),
	$\pi' \MPS \pi$
	holds if and only if there is either (i) systematic ex-ante discrimination against $\pi$ or (ii) no ex-ante discrimination.
\end{lemma}

\begin{proof}
	We shall (implicitly) use \Cref{obs:normalization} throughout.
	For the ``only if'' part, suppose that $\pi' \MPS \pi$.
	Then for any non-empty finite subset $A$ of $\mathbb{R}^\Theta$
	such that $\int_{\mathcal{S}} \max\{v_A,0\} \dd \pi > c$,
	we have
	\begin{equation*}
		\int_{\mathcal{S}} \max\{v_A,0\} \dd \pi'
		= \int_{\mathcal{S}} v_{A \cup \{0\}} \dd \pi'
		\geqslant \int_{\mathcal{S}} v_{A \cup \{0\}} \dd \pi
		= \int_{\mathcal{S}} \max\{v_A,0\} \dd \pi
		> c
	\end{equation*}
	by \hyperref[theorem:Blackwell]{Blackwell's Theorem},
	which is to say that there is either (i) systematic ex-ante discrimination against $\pi$ or (ii) no ex-ante discrimination.

	For the ``if'' direction, suppose that $\pi' \nMPS \pi$.
	Then by \hyperref[theorem:Blackwell]{Blackwell's Theorem}, there is a non-empty finite subset $A$ of $\mathbb{R}^\Theta$ such that
	$\int_{\mathcal{S}} v_A \dd \pi'
	< \int_{\mathcal{S}} v_A \dd \pi$.
	By choosing $\beta \geqslant 0$ large enough, we may ensure that $A' \coloneqq \{ a + \beta : a \in A \}$ satisfies
	\begin{equation*}
		\int_{\mathcal{S}} \max\{v_{A'},0\} \dd \pi'
		< \int_{\mathcal{S}} \max\{v_{A'},0\} \dd \pi .
	\end{equation*}
	Then a $\gamma>0$ may be found such that $A'' \coloneqq \{ \gamma a : a \in A' \}$ satisfies
	\begin{equation*}
		\int_{\mathcal{S}} \max\{v_{A''},0\} \dd \pi'
		\leqslant c
		< \int_{\mathcal{S}} \max\{v_{A''},0\} \dd \pi ,
	\end{equation*}
	so that there is neither (i) systematic ex-ante discrimination against $\pi$ nor (ii) no ex-ante discrimination.
\end{proof}

\begin{remark}\label{remark:lehrer-shmaya}
\Cref{proposition:lehrer-shmaya}'s characterization is valid for any interview cost $c>0$, but does not apply when $c=0$.
In that case, ex-ante discrimination is instead characterized as in \Cref{proposition:lehrer-shmaya}, except with $\MPS$ replaced by the more complete order $N$ defined by $\pi' \mathrel{N} \pi$ iff for every convex, continuous and \emph{non-negative} function $h : \Delta(\Theta) \rightarrow \mathbb{R}$, $\int_{\mathcal{S}}h\dd\pi > 0 $ implies $\int_{\mathcal{S}}h\dd\pi' > 0 $.
This result is due to \textcite[Theorem 2]{LehrerShmaya2008}.
\end{remark}

\renewcommand*{\thesection}{Appendix:}

\section{Proof that \texorpdfstring{$\boldsymbol{\MPS}$}{M} is anti-symmetric}
\label{sec:appendix}

Let $\mathcal{C}$ denote the space of all continuous functions $\Delta(\Theta) \to \mathbb{R}$, equipped with the sup metric.
Write $\mathcal{W} \subset \mathcal{C}$ for those functions that may be written as the difference of two continuous convex functions.

\begin{lemma}
	\label{lemma:denseness}
	$\mathcal{W}$ is dense in $\mathcal{C}$.
\end{lemma}

\begin{proof}
	$\mathcal{W}$ is a vector space since the sum of convex functions is also convex.
	It is also closed under pointwise multiplication \parencite[p. 708]{Hartman1959}, and is thus an algebra.
	In addition, $\mathcal{W}$ contains the constant functions, and it separates points in the sense that for any distinct $s,s' \in \Delta(\Theta)$, there is a $w \in \mathcal{W}$ with $w(s) \neq w(s')$.
	It follows by the Stone--Weierstrass Theorem
	\parencite[see e.g.][Theorem 4.45]{Folland1999}
	that $\mathcal{W}$ is dense in $\mathcal{C}$.
\end{proof}

\begin{proof}[Proof that $\MPS$ is anti-symmetric]
	Suppose that $\pi \MPS \pi' \MPS \pi$; we must show that $\pi = \pi'$.
	The hypothesis means that $\int_{\mathcal{S}} h \dd \pi = \int_{\mathcal{S}} h \dd \pi'$ for every convex and continuous $h: \Delta(\Theta) \rightarrow \mathbb{R}$,
	whence $\int_{\mathcal{S}} w \dd \pi = \int_{\mathcal{S}} w \dd \pi'$ for every $w \in \mathcal{W}$.
	By \Cref{lemma:denseness}, it follows that $\int_{\mathcal{S}} f \dd \pi = \int_{\mathcal{S}} f \dd \pi'$ for every $f \in \mathcal{C}$.

	Thus for any measurable $S \subset \mathcal{S}$,
	by choosing sequence $(f_n)_{n \in \mathbb{N}}$ in $\mathcal{C}$
	such that $\int_{\mathcal{S}} (f_n-1_S) \dd \pi$ and $\int_{\mathcal{S}} (f_n-1_S) \dd \pi'$ vanish as $n \to \infty$,
	we obtain
	\begin{equation*}
		\pi(S)-\pi'(S)
		= \int_{\mathcal{S}} 1_S \dd (\pi-\pi')
		= \lim_{n \to \infty} \int_{\mathcal{S}} f_n \dd (\pi-\pi')
		= 0 .
		\qedhere
	\end{equation*}
\end{proof}







\printbibliography[heading=bibintoc]


\end{document}